\documentclass[12pt]{article}

\usepackage{epsfig}
\usepackage{amssymb}
\usepackage{amsfonts}

\usepackage{color}

\def\be{\begin{equation}}
\def\ee{\end{equation}}
\def\ba{\begin{array}{c}}
\def\ea{\end{array}}

\def\ben{$$}
\def\een{$$}

\newcommand{\bea}{\begin{eqnarray}}
\newcommand{\eea}{\end{eqnarray}}

\newcommand{\kkt}{\kt\!\kt}

\newcommand{\kt}{\rangle}

\newtheorem{thm}{Theorem}

\newtheorem{lemma}[thm]{Lemma}

\newenvironment{proof}{\noindent {\bf Proof}}{\hfill$\square$\vspace{3mm}\endtrivlist}

\begin{document}

\titlepage

 \begin{center}{\Large \bf

Emergence and localization of exceptional points
in an exactly solvable toy model

  }\end{center}


 \begin{center}

\vspace{8mm}

  {\bf Miloslav Znojil} $^{1,2,3,4}$

\end{center}

\vspace{8mm}

  $^{1}$
 The Czech Academy of Sciences,
 Nuclear Physics Institute,
 Hlavn\'{\i} 130,
250 68 \v{R}e\v{z}, Czech Republic, {e-mail: znojil@ujf.cas.cz}

 $^{2}$
 Department of Physics, Faculty of
Science, University of Hradec Kr\'{a}lov\'{e}, Rokitansk\'{e}ho 62,
50003 Hradec Kr\'{a}lov\'{e},
 Czech Republic

$^3$
School for Data Science and Computational Thinking,
Stellenbosch University, 7600 Stellenbosch,
 South Africa

$^4$
Institute of System Science, Durban University of
Technology,
{4001}  Durban,
 South Africa


\section*{Abstract}

In contrast to classical physics,
there are not too many mathematical tools facilitating
the study of
singularities in quantum systems.
One of the exceptions is the Kato's notion
of exceptional points (EPs).
Their emergence and localization are analyzed here
via a family of schematic toy models.


\subsection*{Keywords}.

quasi-Hermitian discrete quantum square well;

boundary-controlled unitary quantum dynamics;

closed formulae for bound-state Sturmians;

non-Hermitian degeneracies {\it alias\,} Kato's exceptional points;

\newpage

\section{Introduction}

Evolutionary singularities
emerging in a classical dynamical system
are a phenomenon which found
its appropriate mathematical clarification and
qualitative classification
in the framework of popular theory of catastrophes
\cite{Zeeman}.
In the majority of applications of the procedure of quantization
people revealed
that there is a deep conceptual difference between
the emergence of singularities in the classical and quantum systems.
One is often
led to conclusion that
there is no immediate quantum analogue of
the theory of catastrophes because
the classical singularity seems {\em always\,}
smeared out after quantization \cite{Messiah}.

The latter
belief
found its particularly persuasive
reconfirmation
in quantum cosmology.
In this field the significant progress
achieved via loop quantum gravity \cite{Ashtekar}
offered a strong support of a replacement of the
point-like
Big Bang  by its
regularized version called Big Bounce
(cf., e.g., the comprehensive monographs \cite{Rovelli,Thiemann}
for details).

We plan to defend our persuasion
that
the situation became radically changed
after the recent innovation
of the formalism of quantum mechanics
using non-Hermitian operators  \cite{book}.
Widely, the innovated theory
became known under the nicknames of
quasi-Hermitian quantum mechanics \cite{Geyer} {\it alias\,}
${\cal PT}-$symmetric quantum mechanics \cite{Carl}
{\it alias} pseudo-Hermitian quantum mechanics \cite{ali}
\textcolor{black}{(see a compact outline of the basic ideas
behind these approaches in Appendix A and, in particular,
in its subsection A.1).}

In the framework of the innovated theory
the apparently unavoidable nature of the regularization
after quantization has been put under question-mark \cite{catast}.
It has been noticed that
the disappearance or survival of singularities
may be model-dependent.
Thus, in particular, one has to conclude that
a strictly quantum Big Bang
can still be treated as a singularity-representing
extreme of a conventional unitary {\em quantum\,} evolution \cite{init}.

In our present paper, the occurrence
and role of some
strictly quantum singularities
will be discussed. They will be interpreted as
an inseparable part of
a remarkable non-Hermitian (or rather quasi-Hermitian) collapse
or, in opposite direction, of
another specific process of a
non-Hermitian singularity unfolding.

For the sake of definiteness, a
schematic model
will be only considered.
It will be
shown to
exhibit a number of counterintuitive features.
In particular, we will emphasize that
a pair of some of its
neighboring bound or resonant states
may merge
at a value of
a parameter called
exceptional point (EP, cf. \cite{Kato}
or \textcolor{black}{the subsection A.2 of}
Appendix A below).

In contrast to several rather sceptical conclusions
about the model
as reached in our recent
contribution to conference proceedings \cite{init},
we will be able to report
a significant progress
in our understanding of the underlying quantum dynamics.
In particular, in several benchmark special cases
we will be able to
prove the existence of the EP singularities which will appear
localizable
non-numerically.

\section{The model}

The kinetic energy of a quantum particle
which moves freely
along an equidistant
1D lattice is represented by a discrete Laplacean \textcolor{black}{\cite{grid}}.
In conventional textbooks
such a motion is often studied as
restricted
to a finite segment of the lattice,
with the most common
Dirichlet boundary conditions
imposed at its ends.
The energy levels
can be then found as eigenvalues of
Hermitian quantum $N$-by-$N$-matrix Hamiltonian
 \be
 H_{}^{(N)}=
 \left[ \begin {array}{ccccc}
  2&-1&0
 &\ldots&0
 \\
 \noalign{\medskip}-1&2&-1&\ddots&\vdots
 \\
 \noalign{\medskip}0&-1&\ddots&\ddots
 &0
 \\
 \noalign{\medskip}\vdots&\ddots&\ddots&2&-1
 \\
 \noalign{\medskip}0&\ldots&0&-1&2
 \end {array} \right]\,.
 \label{heKat}
 \ee
The model is exactly solvable because its eigenvectors
can be sought in the form of superposition of
classical Tschebyshev
polynomials \textcolor{black}{\cite{Ryshik}}.

In our older paper \cite{[17]}
we revealed that the latter form of solvability
survives
a \textcolor{black}{certain} generalization.
\textcolor{black}{The essence of the generalization consists in a}
transition to the
boundary-controlled parameter-dependent model
\textcolor{black}{and Hamiltonian}
 \be
 H_{}^{(N)}(z)=
 \left[ \begin {array}{ccccc}
  2-z&-1&0
 &\ldots&0
 \\
 \noalign{\medskip}-1&2&-1&\ddots&\vdots
 \\
 \noalign{\medskip}0&-1&\ddots&\ddots
 &0
 \\
 \noalign{\medskip}\vdots&\ddots&\ddots&2&-1
 \\
 \noalign{\medskip}0&\ldots&0&-1&2- z^*
 \end {array} \right]\,
  \,
 \label{usKa8t}
 \ee
in which the parameter \textcolor{black}{itself} can be complex,
\textcolor{black}{$z \in {\mathbb{C}}$:}
\textcolor{black}{A few other related technical details
can be found summarized in Appendix C below}.

Surprisingly enough,
even the unconventional and manifestly non-Hermitian
$N-$site lattice
version of
such a quantum square well model with $z \notin {\mathbb{R}}$
(i.e., in effect,  with
the complex Robin boundary conditions, \textcolor{black}{cf.
Appendix C})
can be attributed a more or less conventional physical
probabilistic interpretation.
Indeed, \textcolor{black}{in
\cite{[17]}}
we managed to show that there exists a non-empty
complex domain ${\cal D}$ of parameters $z$
at which the spectrum of the model remains
strictly real and non-degenerate.

The
operator
can serve, therefore,
as an exactly solvable stationary
toy-model Hamiltonian
fitting the conventional quantum mechanics of unitary systems
in its recent  quasi-Hermitian reformulation
(cf. \cite{Geyer} and also \cite{book,Carl,ali,Dyson,Dieudonne}).
The manifest non-Hermiticity of matrix $H_{}^{(N)}(z)$ with complex
$z\in {\cal D}$,
reflects merely the fact that our conventional tacit acceptance
of
the most common $N-$dimensional Hilbert space
$\mathcal{H}^{(N)}_{mathematical}=\mathbb{C}^N$
is unphysical.
In a way recalled in Appendix A,
its necessary
conversion into another, acceptable physical
Hilbert space $\mathcal{H}^{(N)}_{physical}$
is more or less straightforward.

The goal
is to be achieved via an amended inner-product metric $\Theta$.
The details
of the underlying theory can be found
explained in \cite{Geyer}
or in Appendix B below.
On these grounds
one can conclude that for the evolution which is
generated by a preselected non-Hermitian but quasi-Hermitian Hamiltonian
$H\neq H^\dagger$
with real spectrum, the unitarity can be guaranteed
by the condition,
 \be
 H^\dagger\Theta^{}=\Theta^{}\,H\,,
 \label{dieudoce}
 \ee
i.e., by the
Dieudonn\'{e}'s \cite{Dieudonne}
quasi-Hermiticity property of $H$.

The assignment $H \to  \Theta$
is not unique.
In applications,
the construction of one or more operators $\Theta$
may
represent a decisive
technical challenge
(see, e.g., an extensive
review of this item in \cite{ali}).
For the stationary model (\ref{usKa8t}), therefore,
such an assignment has been performed,
in \cite{[17]},
in explicit manner. A brute-force
solution of the finite set of
the $N^2$ algebraic equations (\ref{dieudoce})
for the unknown matrix elements of $\Theta$
has been used for the purpose.

The
demonstration of feasibility of the assignment $H \to  \Theta$
reconfirmed
the
appeal
of quantum mechanics in its stationary
quasi-Hermitian formulation of reviews \cite{Geyer,ali,Brody}.
Incidentally, the practical use of the formalism becomes
much more technically complicated when one
omits the condition of stationarity.
Still, a
full conceptual
consistency of quantum mechanics in its non-stationary
quasi-Hermitian formulation
can be achieved (cf. \cite{Faria,timedep,SIGMA,Fring,NIP,Android}).

Along the latter lines,
the constructions based on the non-stationary
and
non-Hermitian observable
Hamiltonians
remained difficult but still feasible.
Recently, fresh
developments in the field were
initiated
by Fring et al \cite{[34]} and, independently,
by Matrasulov et al
\cite{Matrasulov}.
In both of these collaborations
it has been clarified that
it will make good sense to extend
the applications of the quasi-Hermitian
quantum mechanics to
the unitary quantum systems in which the quasi-Hermitian quantum
observables
become manifestly
time-dependent.

The new optimism has also been advocated
in our recent study \cite{PS1} where we decided to replace
the stationary solvable toy-model of
Eq.~(\ref{usKa8t})
by its
non-stationary generalization containing
a nontrivial, non-constant complex function of time $z=z(t)$.
Still, a number of questions remained unanswered (cf.
their presentation \cite{init} during a last-year conference).
And precisely this
survival of open questions
also motivated our
present return to
the model
and to its upgraded analysis, with the main attention shifted
to the study of its genuine quantum singularities.

\section{Specific features of boundary-controlled dynamics
}


One of the best visible and phenomenologically most
deplorable gaps
in our understanding of both the stationary
and non-stationary versions
of model (\ref{usKa8t})
can be seen in the
questions
concerning the
existence and, if they do exist, the localization
of its singularities.
These questions remained unanswered
in \cite{PS1}. Moreover, only a very few
concise
answers were provided later in \cite{init}.

In the latter study, indeed,
the questions concerning
the genuine quantum EPs caused by boundary conditions
have only been addressed via several numerical illustrative examples.
The reason was
not only the lack of non-numerical insight but also
the lack of space  as provided  by the proceedings.
Both of these shortcomings
appeared decisive.

More recently
we returned to the problem, and we
arrived at a much better and predominantly non-numerical
understanding of the
role and structure of singularities. Thus, we are now going
to complement the key messages of
\cite{init} and to enrich and
enhance this note to a full-paper format.

\subsection{The even$-N$ models}

Due to a certain favorable hidden symmetry of
matrices (\ref{usKa8t}) with a purely imaginary $z(t)$
the localization of EPs
appeared
comparatively easy at
the even matrix dimensions $N=2,4,\ldots$.
After a reparametrization of $z(t)=i\, \sqrt{1-r^2(t)}$
we also found, in \cite{[17]}, that
it makes sense to treat the new variable $r(t)$
as a real and, say, non-decreasing function of time.

In {\it loc. cit.} we mentioned two
main consequences of the reparametrization.
First, the model only proved non-Hermitian
(i.e., of our methodical interest)
at $r^2(t) \leq 1$.
Second,
the spectrum of $H^{(N)}(t)$
remains smoothly time-dependent
and real at all of the real parameters $r(t) \in \mathbb{R}$.
Thus, the non-Hermiticity -- Hermiticity quantum
phase transition (cf. \cite{[34b]}) is smooth. Here,
the phenomenon can be found illustrated in Figure \ref{ufigone}
where we choose $N=6$.

\begin{figure}[h]
\begin{center}
\epsfig{file=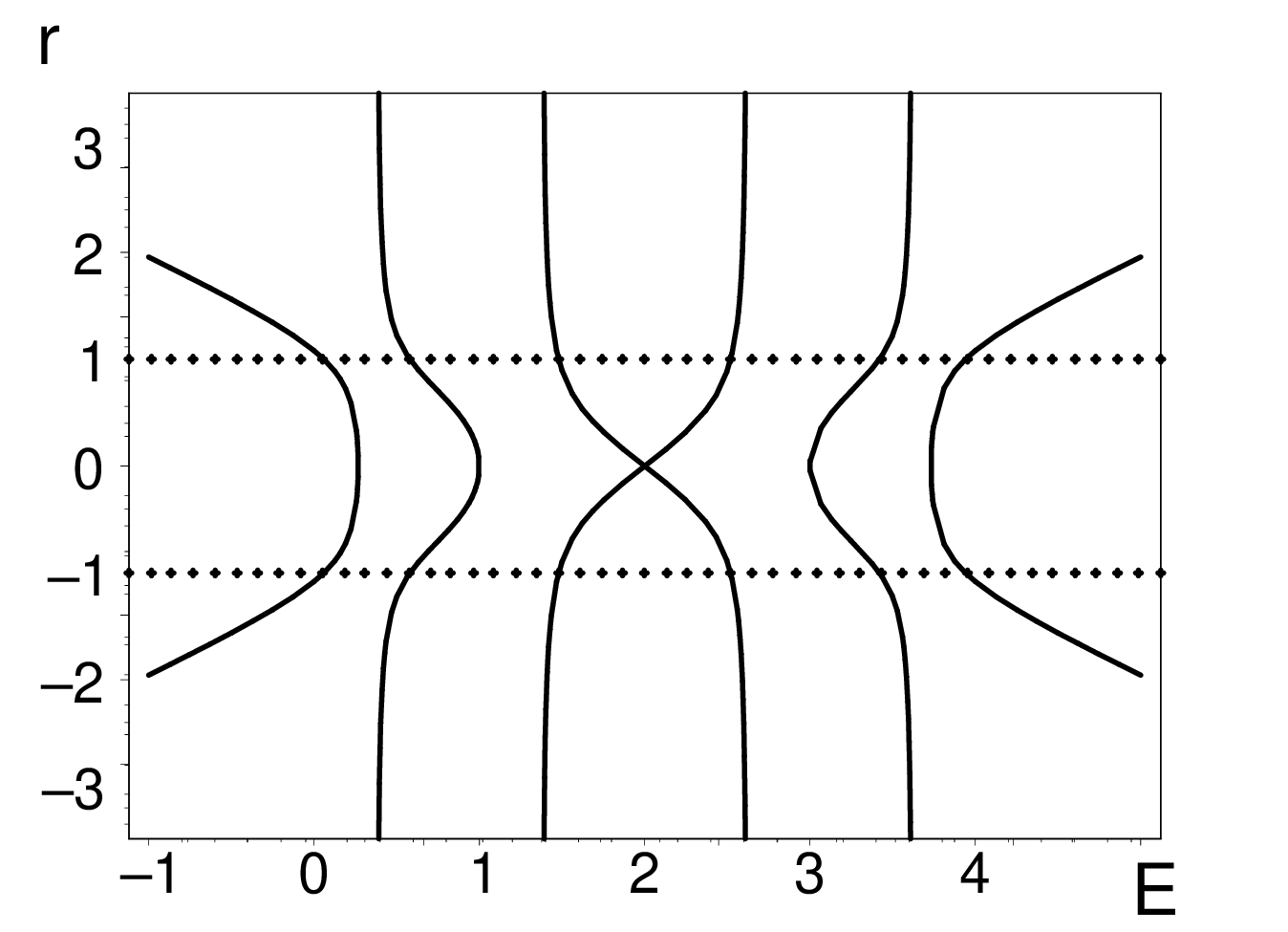,angle=270,width=0.35\textwidth}
\end{center}
\vspace{-2mm}\caption{Parameter $r$ versus energy $E$
for $z=i\, \sqrt{1-r^2}$
at $N=6$. Two auxiliary \textcolor{black}{dotted} lines of $r=\pm 1$
mark the boundary of the non-Hermiticity of matrix $H^{(N)}(z)$
of Eq.~(\ref{usKa8t}).
 \label{ufigone}}
\end{figure}

In a subsequent commentary \cite{init}
we emphasized that at any even $N$
the spectrum
remains discrete and non-degenerate, first of all,
in
the standard
Hermitian regime (i.e., at $r^2>1$) . In contrast,
in a way well visible also in Figure \ref{ufigone} here,
the loss of
the manifest Hermiticity at $r^2 \leq 1$
has been mentioned to open
the possibility of a degeneracy
of a pair of
energy levels in the maximal non-Hermiticity limit
of $r \to 0$.

Besides a numerical demonstration of
these results, and besides
their graphical representations, it
would be also desirable to prove them analytically.
Indeed, only then one can conclude that
at any even dimension $N$
we always encounter,
at $r=0$, a genuine
non-Hermitian degeneracy.

\subsection{Odd$-N$ problem}

There were several reasons why we failed
to describe the exceptional point singularities
at the odd matrix dimensions $N = 3, 5, \ldots$
in \cite{init}.
In what follows, these cases will appear to be tractable,
first of all, thanks to a simplification of the task.
Trivial as it may look, it will consist
in an elementary shift of the
energy scale ($E \to E-2$)
so that our toy-model Hamiltonian
will acquire its perceivably
more transparent equivalent matrix form
 \be
 H_{}^{(N)}(t)=
 \left[ \begin {array}{ccccc}
  -z(t)&-1&0
 &\ldots&0
 \\
 \noalign{\medskip}-1&0&-1&\ddots&\vdots
 \\
 \noalign{\medskip}0&-1&\ddots&\ddots
 &0
 \\
 \noalign{\medskip}\vdots&\ddots&\ddots&0&-1
 \\
 \noalign{\medskip}0&\ldots&0&-1&- z^*(t)
 \end {array} \right]\,.
 \label{Ka8t}
 \ee
This will enable us to see that
at odd $N$
there exists an anomalous
bound-state-energy
root $E=0$
of the secular equation
which is $r-$independent.
For this reason
one can immediately deduce that there is no level-crossing
EP degeneracy at $r=0$.

In \cite{init},
the latter observation forced us to add a
non-vanishing real part to $z(t)$. This,
unfortunately, made the secular equation so
complicated that we
had to resort, in the major part of paper \cite{init},
to the mere purely numerical
study of our toy model (\ref{usKa8t})
(see also a concise summary of these efforts in Appendix C below).

This was a technical complication which
certainly limited the appeal of our results.
Due to these obstacles
we only managed to describe and
understand the mechanism of the emergence of
the EP singularity, via two illustrative pictures,
just
at the first nontrivial choice of dimension $N=5$.
Later on, a remarkable progress has been achieved in the
non-numerical forms to be
reported in the present paper.



\section{Solvability: Sturmians}

Although  the spectrum of Figure \ref{ufigone}
as assigned to matrix (\ref{usKa8t})
at $N=6$ does not seem to be too complicated,
we did not mange to reveal, in it, any traces of
the symmetries
as observed in the picture.
In \cite{init} we still conjectured that
the formal core of the feasibility of the
localization of the EPs at $N=6$
has to be seen in the user-friendly
structure of the related secular equation.

In {\it loc. cit.}
there was no space to
make the argument explicit, and to
support the claim
by the formulae. This is to be
done in what follows.

\begin{table}[th]
\caption{Sturmian solutions  of
secular equations for the present simplified model (\ref{Ka8t}).
 } \label{pexp3b}
\begin{center}
\begin{tabular}{||c|c||}
\hline \hline
  $N$ &
  {$r^2(E^2)$ }\\
 \hline \hline
 2&
 $E^2$
 \\
 3&
  $E^2-1$
 \\
 4&
 $E^2\,(E^2-2)/(E^2-1)$
 \\
 5&
 $(E^4-3\,E^2+1)/(E^2-2)$
 \\
 6&
 $E^2\,(E^2-1)\,(E^2-3)/(E^4-3\,E^2+1)$
 \\
 7&
 $(E^6-5\,E^4+6\,E^2-1)/[(E^2-1)(E^2-3)]$
 \\
 8&
 $E^2\,(E^6-6\,E^4+10\,E^2-4)/(E^6-5\,E^4+6\,E^2-1)$
 \\
 9&
 $(E^8-7\,E^6+15\,E^4-10\,E^2+1)/(E^6-6\,E^4+10\,E^2-4)$
 \\
  \hline \hline
\end{tabular}
\end{center}
\end{table}

\subsection{Non-numerical localizations of EPs}

The formulae of paper \cite{init}
become almost miraculously simplified
after the transition to the shifted-scale model (\ref{Ka8t}).
This can be found demonstrated
in our present Tables~\ref{pexp3b} and~\ref{pexp3a}.
We see there, in particular, that the upgraded $N=6$ item is
much more compact and transparent than its unshifted-scale
predecessor of paper \cite{init}.
Also the existence of certain additional
parity-symmetry breaking sub-factorizations of
Sturmian \cite{Sturm}
couplings $r^2(E^2)$
as sampled in Table \ref{pexp3a} appeared
not only equally unexpected
but also fairly useful, especially
for our present purposes
of the localization of the EPs
(see the details below).

The
manifest $E \to -E$
symmetry of the Sturmians $r^2=r^2(E^2)$ is
visible also in Figure~\ref{ufigone}.
This is a benefit
of the model which remains visible even
after
the factorization
of the formulae
as displayed
in Table~\ref{pexp3b}.
Serendipitously, one reveals
also another, ``hidden'' symmetry
of these results
by which,
up to a certain $E^2-$factor anomaly, {\em all\,} of
the factors found in the numerators at some $N$
are found relocated into denominators at $N+1$.


\begin{table}[th]
\caption{Sample of further auxiliary factorizations \label{pexp3a}}
\begin{center}
\begin{tabular}{||c|c||}
\hline \hline
  N &  denominator of {$r^2(E^2)$ }\\
 \hline \hline
 6&$ E^4-3\,E^2+1=(E^2-1+E)(E^2-1-E)$
 \\
 8&$ E^6-5\,E^4+6\,E^2-1 = (E^3-2\,E+E^2-1) (E^3-2\,E-E^2+1)$
 \\
 9&$ E^6-6\,E^4+10\,E^2-4=(E^2-2)(E^4-4\,E^2+2) $
 \\
  \hline \hline
\end{tabular}
\end{center}
\end{table}


%
%

In \cite{init}, one of
our main goals was to prove that the EP degeneracy
can be localized even when the
matrix dimension $N$ is odd.
Unfortunately, the task remained unfulfilled.
Indeed,
we only managed to
explain that one has to use
a shifted complex form of
parameter $z(t)=y(t)+{\rm i}\,\sqrt{1-r^2(t)}$ containing
a non-vanishing constant or time-dependent
real shift $y(t) \neq 0$.

We only found that
the central level crossing as presented in
Figure~\ref{ufigone} at $N=6$ disappears at odd $N$.
For illustration we choose $N=5$
and used a purely numerical approach
-- see a compact outline of the argumentation in Appendix C below.
Now, these results will be complemented
by their study using rigorous analytic techniques.

\subsection{{Example}}

For the purposes of our present search of the EPs
the role of the shift $u$ in
 \be
 z=-u+{\rm i}\,\sqrt{1-r^2}
 \ee
is trivial at
$N=2$. Its change just moves the
origin of the energy scale.
This means that
it is sufficient to confirm
the existence of the
EP singularity at $u=E=0$ (cf. the first line in Table \ref{pexp3b}).

This is an elementary but explicit confirmation of the
existence of a singularity
tractable as the Kato's exceptional point.
This is a mathematically rigorous result which is
an immediate consequence of the
following elementary observation and construction.

\begin{lemma}.
Matrix
 \be
 H^{(2)}_{(EP)}(u)=
 \left[ \begin {array}{cc}u -i&-1\\
 \noalign{\medskip}-1&u+i\end {array} \right]
 \ee
is not diagonalizable. It can only be
given the canonical Jordan form via
a matrix-intertwining relation
 \be
 H^{(2)}_{(EP)}\,Q^{(2)}_{(EP)}=
 Q^{(2)}_{(EP)}\,
 \left[ \begin {array}{cc} u&1\\
 \noalign{\medskip}0&u\end {array} \right]
  \label{11}
 \ee
where the invertible intertwiner
 \be
 Q^{(2)}_{(EP)}=
 \left[ \begin {array}{cc} -i&1\\\noalign{\medskip}-1&0\end {array} \right]
 \ee
is usually called transition matrix.
\end{lemma}

\section{Analytically solvable benchmark models}

\subsection{The first nontrivial model: $N=3$}

Even the choice of matrix dimension as small as
$N=3$ makes
our insight
in the spectrum perceivably worsened.
The reason is that the necessary
Cardano formulae yielding the energies are
far from nice.

\begin{figure}[h]
\begin{center}
\epsfig{file=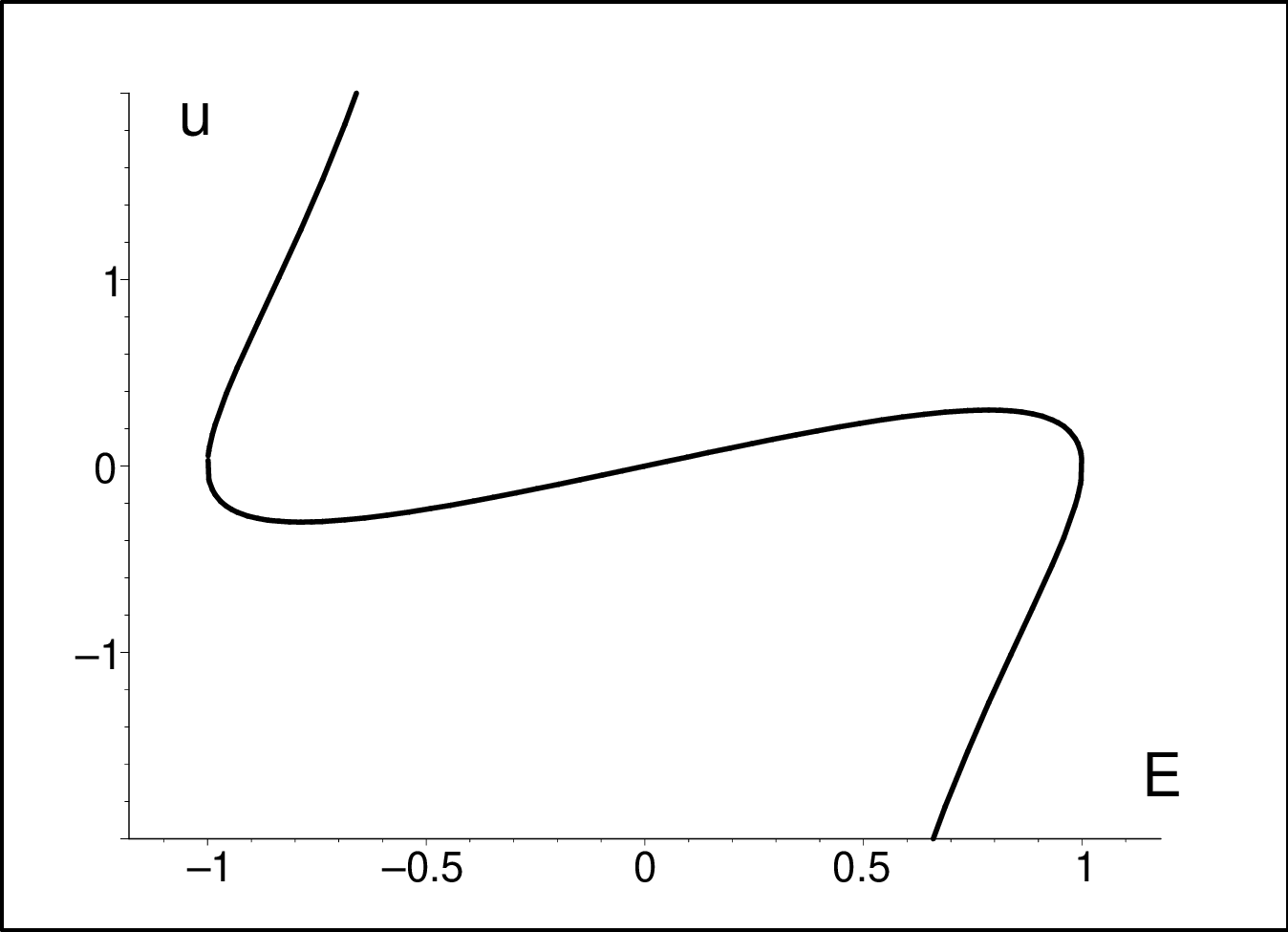,angle=270,width=0.3\textwidth}
\end{center}
\vspace{-2mm}\caption{Graph of the Sturmianic curve
$\textcolor{black}{u= \ }  u(E)$ at $N=3$.
 \label{3drufi}}
\end{figure}

The $r \to -r$ symmetry of the spectrum
(usable, after all, at any $N$)
enables us to deduce that
the deformation of the spectral curves
as caused by the changes of the shift $u$
can only lead to a degeneracy of some levels at
$r=0$.
Hence, it is sufficient to study the spectrum of
matrix
 \be
 H^{(3)}(u)=
 \left[ \begin {array}{ccc} u-i&-1&0\\
 \noalign{\medskip}-1&0&-1\\\noalign{\medskip}0&-1&u+i
 \end {array} \right]
 \label{10}
 \ee
i.e., the roots of its characteristic polynomial
 \be
 P(u,E)=
 {{\it E}}^{3}-2\,u{{\it E}}^{2}- \left( 1-{u}^{2} \right)
 {\it E}+2\,u\,.
  \label{14}
 \ee
An intuitive insight in the form of the spectrum
is provided, in Sturmian representation,
by Figure~\ref{3drufi}. In this
picture we see that a pairwise confluence of
the levels can only be achieved at the minimum or maximum of the
bounded part
of the curve given by one of the Sturmian roots
of equation $P(u,E)=0$, viz., by formula
 \be
  u(E)={\frac {{{\it E}}^{2}-1+\sqrt {-{{\it E}}^{2}+1}}{{\it E
  }}}\,.
  \label{16}
 \ee
For the rigorous proof of the fact that the confluence
of the energies is of the Kato's type, i.e., that it is
accompanied also by the confluence
of the respective eigenvectors,
it is again sufficient to construct the transition matrix
and/or to prove the non-diagonalizability of the Hamiltonian.

\begin{lemma}.
Matrix
 \be
\left[ \begin {array}{ccc} -1/4\,
\left (\sqrt {-2+2\,\sqrt {5}}\right )^{3}&0&0\\
\noalign{\medskip}0&1/2\,\sqrt {-2+2\,
\sqrt {5}}&1\\\noalign{\medskip}0&0&1/2\,\sqrt {-2+2\,\sqrt {5}}
\end {array} \right]
 \ee
is the Jordan-block representation of the toy model (\ref{10}) at
its right EP singularity.
\end{lemma}
\begin{proof} is straightforward and
its short version could
proceed just by insertion
in the $N=3$ analogue of Eq.~(\ref{11}).
What led to the result
was the exact and unique specification (\ref{16}) of the
root of the characteristic polynomial.
The (say, positive) maximum of function (\ref{16})
appeared then determined by the standard rule
$u'(E)=0$, i.e., by the
cubic equation
for $x=E^2$,
  \be
  (1-x)(1+x)^2=1
  \ee
possessing the unique positive closed-form solution
 \be
 x=1/2\,\sqrt{5} -1/2 \approx 0.6180339887\,.
 \ee
The
EP energy
$E \approx 0.7861513775$
is related to the reconstructed EP-supporting shift
 \ben
 u^{(EP)}=1/2\,\sqrt {-2+2\,\sqrt {5}}
 -2\,{\frac {1}{\sqrt {-2+2\,\sqrt {5}}}}+
\sqrt {4\, \left( -2+2\,\sqrt {5} \right) ^{-1}-1}
 \ \approx \ 0.3002831061\,.
 \een
\end{proof}


\subsection{Benchmark model with $N=4$}

The same procedure can be applied to the
$N=4$ toy-model-Hamiltonian
matrix
 \be
 H^{(4)}(u)=
 \left[ \begin {array}{cccc} u-i&-1&0&0\\\noalign{\medskip}-1&0&-1&0
\\\noalign{\medskip}0&-1&0&-1\\\noalign{\medskip}0&0&-1&u+i
\end {array} \right]
 \ee
yielding the
secular polynomial of the fourth order in the energy,
 \be
 E^4-2\,u\,E^3+(-2+u^2)\,E^2+4\,u\,E-u^2\,.
 \ee
Its form is compatible with the existence of the
trivial EP singularity
at $E=0$ and $u=0$.

\begin{figure}[h]
\begin{center}
\epsfig{file=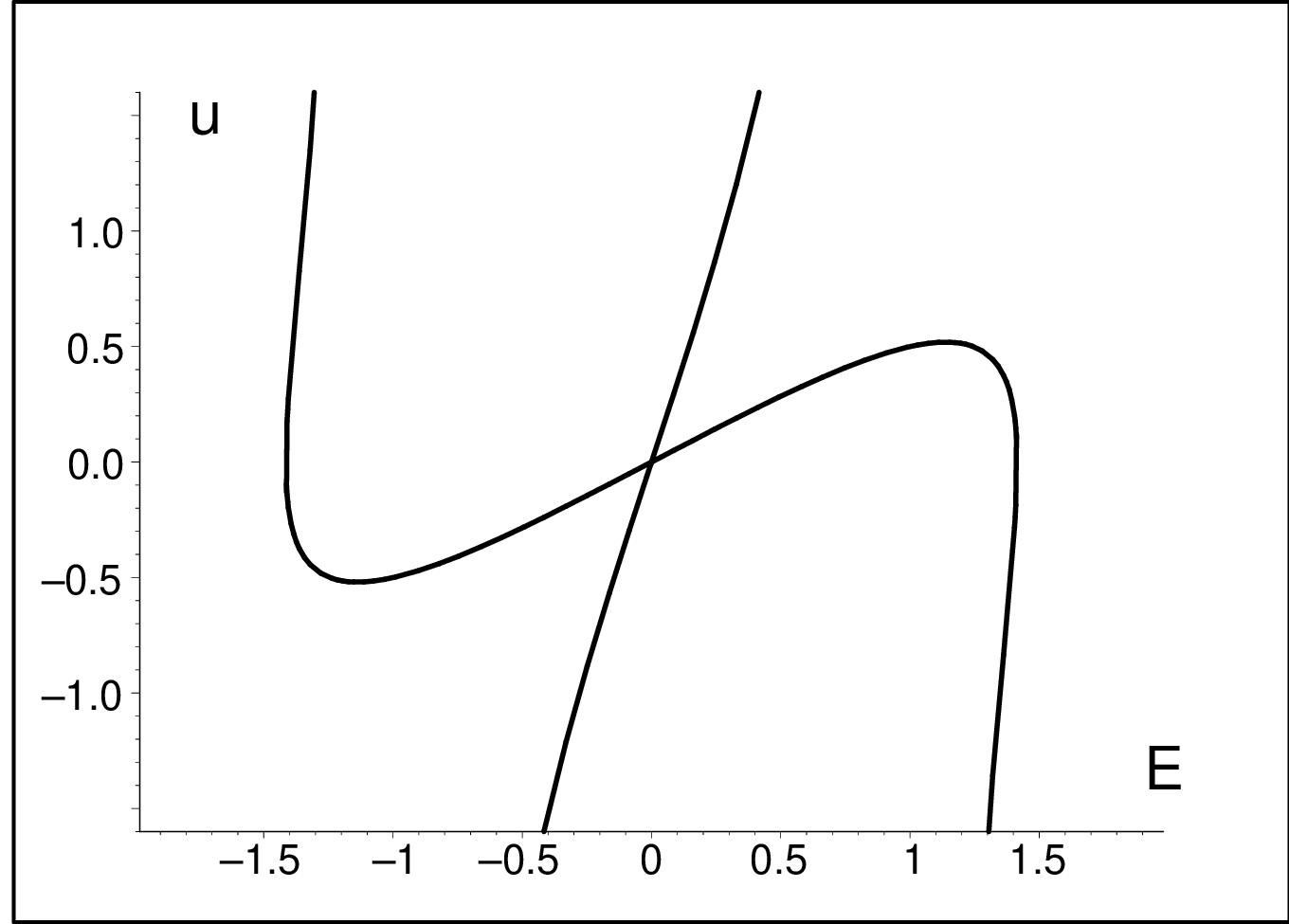,angle=270,width=0.3\textwidth}
\end{center}
\vspace{-2mm}\caption{$u(E)$ for $N=4$.
 \label{4drufi}}
\end{figure}

A nontrivial task can be now formulated as the question and proof
of existence of the other,
``off-central'' EP  singularity or singularities at
some nontrivial shift or shifts $u \neq 0$.

A non-rigorous answer is provided by Figure \ref{4drufi}
in which we see that in a close parallel with the preceding
case of $N=3$, also the  $N=4$ Sturmian curve
 \be
 u(E)= E+
 {\frac {
 \sqrt {2-{{\it E}}^{2}}-1
 } {{{\it E}}^{2}-1}}\,{E}\,
 \ee
has the two off-central $u \neq 0$
extremes indicating the emergence of the EPs.

\begin{figure}[h]
\begin{center}
\epsfig{file=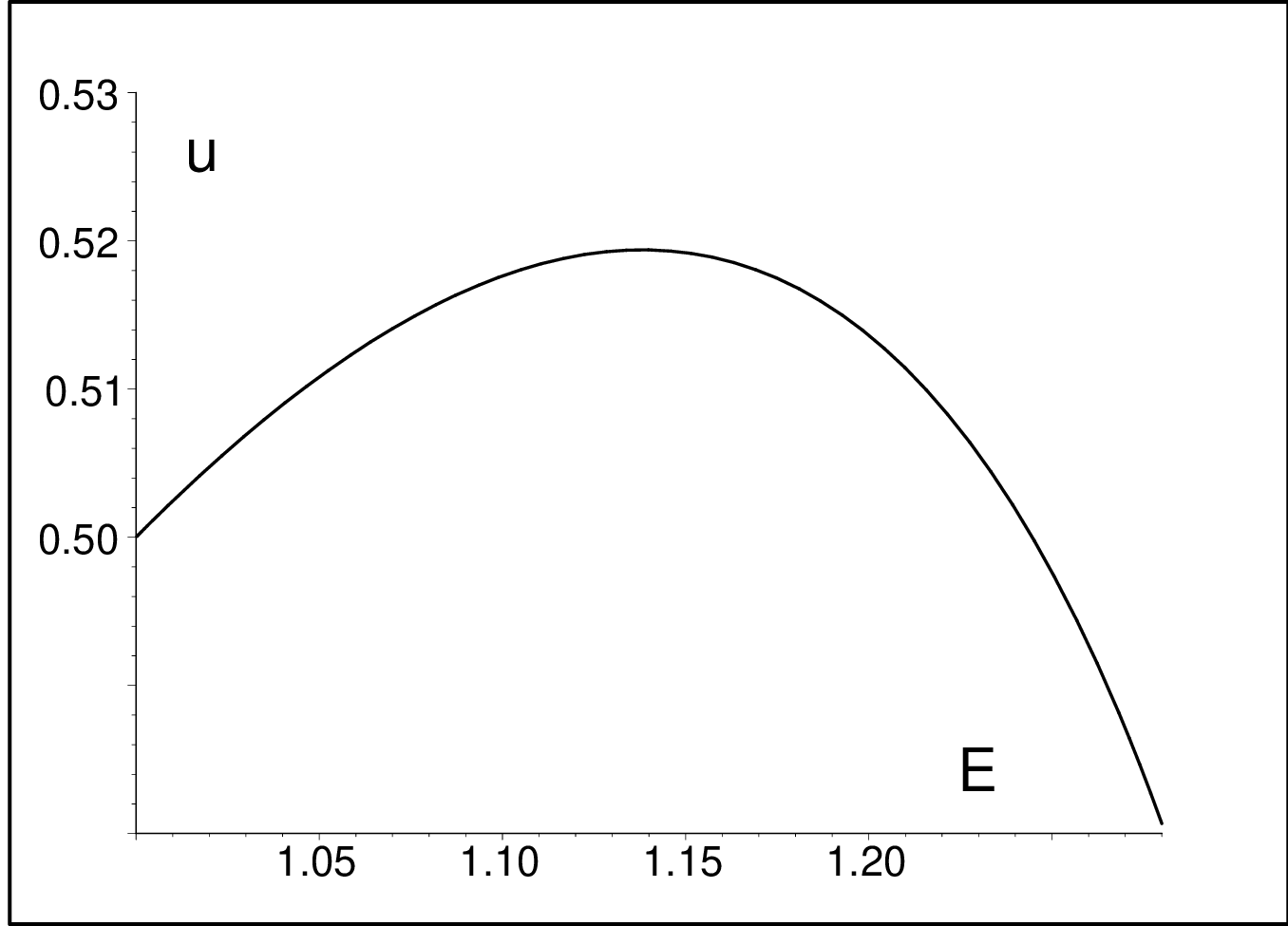,angle=270,width=0.3\textwidth}
\end{center}
\vspace{-2mm}\caption{Graphical localization of the
EP-determining maximum of $u(E)$ at $N=4$.
 \label{4bdrufi}}
\end{figure}

For any practical purposes it is sufficient to
localize the off-central-EP coordinates
$u^{(EP)}$ and $E^{(EP)}$ approximatively,
using a  suitable magnification of the graph
of Figure \ref{4drufi}
(see
Figure \ref{4bdrufi} as a sample
of such a magnification
and graphical localization).
Nevertheless, the rigorous answer
is also accessible.
Along the same lines as above, it
can be obtained
in a comparatively compact form of expression
 \be
 E^{(EP)}=
 1/3\,\sqrt {3\,\sqrt [3]{26+6\,\sqrt {33}}-24\,{\frac {1}{\sqrt [3]{26
 +6\,\sqrt {33}}}}+6}
 \ \approx \ 1.138243270\,.
 \ee
We can conclude that this result
is fully compatible with the graphical solution
as shown in Figure \ref{4bdrufi}.

\section{\textcolor{black}{Beyond $N=4$}}

\subsection{\textcolor{black}{Odd versus even matrix dimensions $N$}}

\textcolor{black}{Comparison of Figures \ref{3drufi}
and \ref{4drufi} reveals a certain
intimate qualitative
correspondence between the positions of the
EP singularities
at
$N=3$ and $N=4$.
Indeed,
the information
about the EPs
which is carried by the
function $u(E)$
at
$N=3$
only differs
from
the information
about the EPs
at
$N=4$
by the emergence of an additional, third, trivial EP at $E=0$.}

\textcolor{black}{It is, naturally,
tempting to conjecture
that
such a correspondence might be extensible
to any larger pair of neighboring matrix dimensions
$N=2k-1$ and $N=2k$.
This, really, opens the possibility
of the generalization of our
results to the models with $k = 3, 4, \ldots$
and, in principle, even with the very large
pairs of
matrix dimensions with
$k \gg 1$.}

\textcolor{black}{For a verification
of the validity and of the possible
explicit forms of such a type of
conjecture
one can feel encouraged
by the survival of simplicity of the corresponding
general Hamiltonians at $r=0$,
 \be
 H^{(N)}(u)=
 \left[ \begin {array}{cccccc}
  u-i&-1&0&\ldots&0&0
 \\
 \noalign{\medskip}-1&0&-1&0&\ldots&0
 \\
 \noalign{\medskip}0&-1&0&\ddots&\ddots&\vdots
 \\
 \noalign{\medskip}0&0&\ddots&\ddots&-1&0
 \\
 \noalign{\medskip}\vdots&\ddots&\ddots&-1&0&-1
 \\
 \noalign{\medskip}0&\ldots&0&0&-1&u+i
 \end {array} \right]
 \ee
yielding the related explicit forms of the
secular polynomials in a more or less routine manner
(the task is left to the readers).}

\subsection{The $N=5$ model revisited}

\textcolor{black}{A verification
of the latter
conjecture has to start
in the first truly nontrivial model with $k=3$
and odd $N=5$.}
The purely numerical analysis of the spectrum of \textcolor{black}{such}
a
``generic odd$-N$'' example of our toy
model~(\ref{usKa8t})
was performed in \cite{init}.
In the light
of Table~\ref{pexp3b}
as well as in the light
of our preceding,
purely analytic description of the
``generic even$-N$''
model with $N=4$
it is possible to expect that
a \textcolor{black}{certain} increase of the complexity of
Sturmians $r^2(E^2)$
\textcolor{black}{would already} enter the scene at $N=6$.

\begin{table}[th]
\caption{Visualization-friendly re-arrangements of some
formulae of Table \ref{pexp3b}}
\begin{center}
\begin{tabular}{||c|c||}
\hline \hline
  N & $r^2(E^2)$ \\
 \hline \hline
 4&
 $E^2-1-1/(E^2-1)$
 \\
 5&
 $E^2-1-1/(E^2-2)$
 \\
 6&
 $E^2-1-(E^2-1)/(E^4-3\,E^2+1)$
 \\
 & $=E^2-1-1/(E^2-2-1/(E^2-1))$
  \\
   \hline \hline
\end{tabular}
\end{center}\label{qexp3a}
\end{table}

This expectation can be
further supported by Table~\ref{qexp3a}
in which
we display certain
partial simplifications
of the Sturmians $r^2(E^2)$ at $N=4$, $N=5$ and $N=6$.
Thus,
along the same methodical lines as used above, the
basic orientation in the structure
and parameter-dependence of the $N=5$ spectrum
can be obtained in full analogy with its $N=4$ predecessor.

\begin{figure}[h]
\begin{center}
\epsfig{file=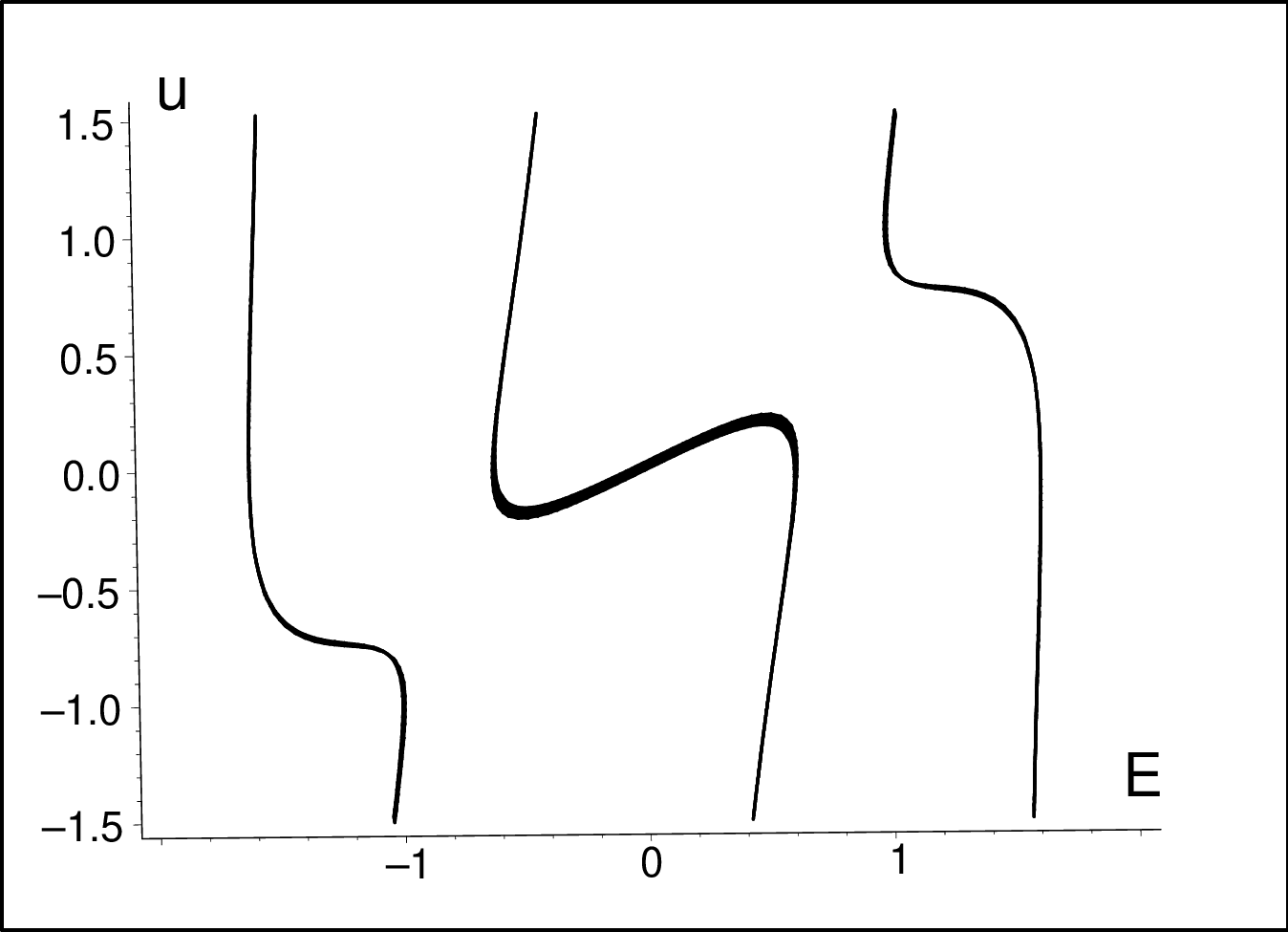,angle=270,width=0.3\textwidth}
\end{center}
\vspace{-2mm}\caption{$u(E)$ for $N=5$.
 \label{5bdrufi}}
\end{figure}

What is to be expected is
the emergence of the two off-central non-Hermitian
EP degeneracies at \textcolor{black}{$r=0$ and at}
some two critical shifts
$u^{(EP)}_{(\pm)}=\pm |u^{(EP)}_{(\pm)}|$.
The expectation is fully confirmed by the numerical
experiments of \cite{init} as well as by our new
numerically generated Figure~\ref{5bdrufi}.
Its inspection reveals that
the interval of $u$ inside which the whole
$N=5$ spectrum remains real
is rather small.

\begin{figure}[h]
\begin{center}
\epsfig{file=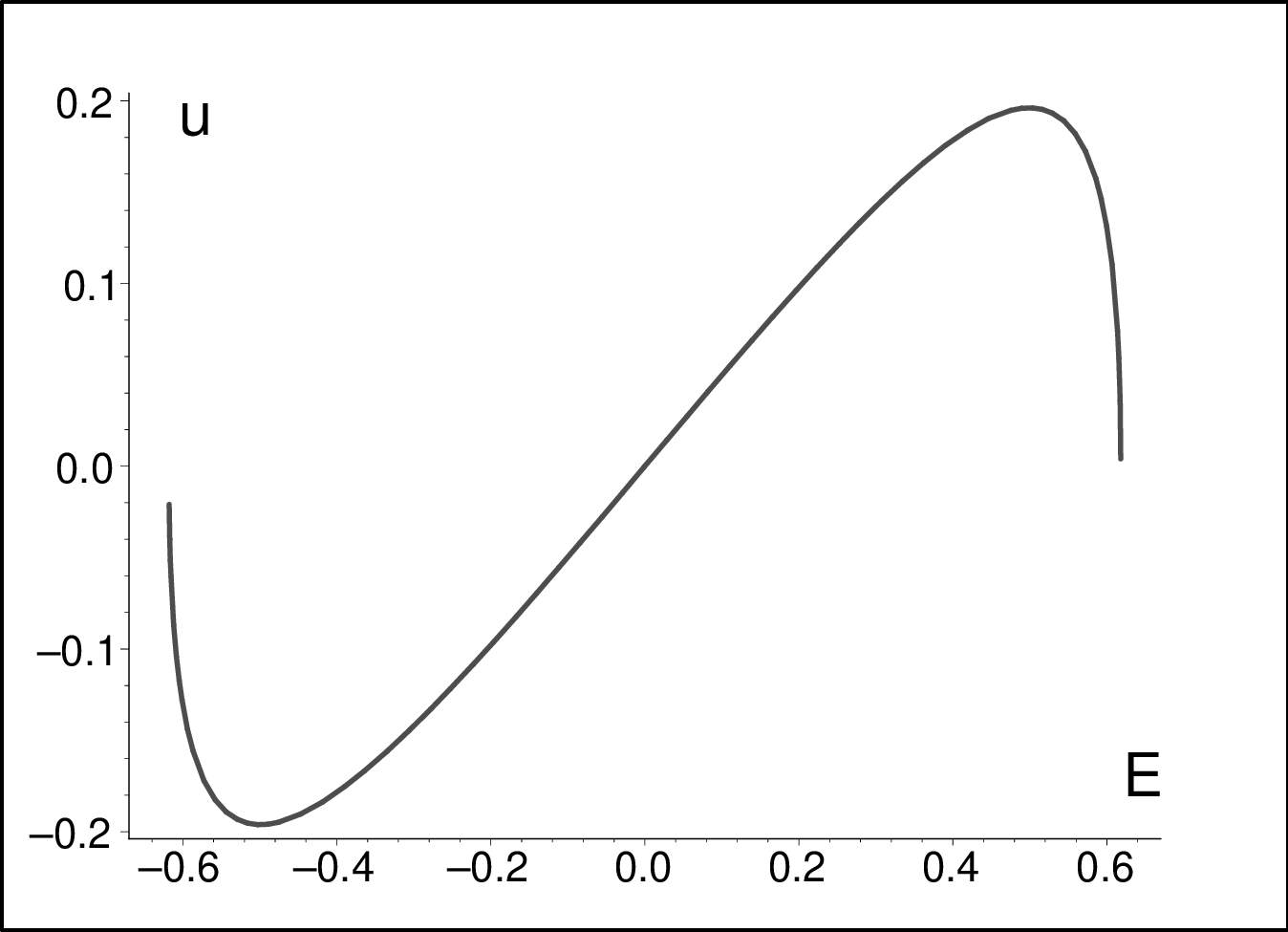,angle=270,width=0.3\textwidth}
\end{center}
\vspace{-2mm}\caption{$u(\textcolor{black}{E})$ for $N=5$.
 \label{5drufi}}
\end{figure}

\begin{figure}[h]
\begin{center}
\epsfig{file=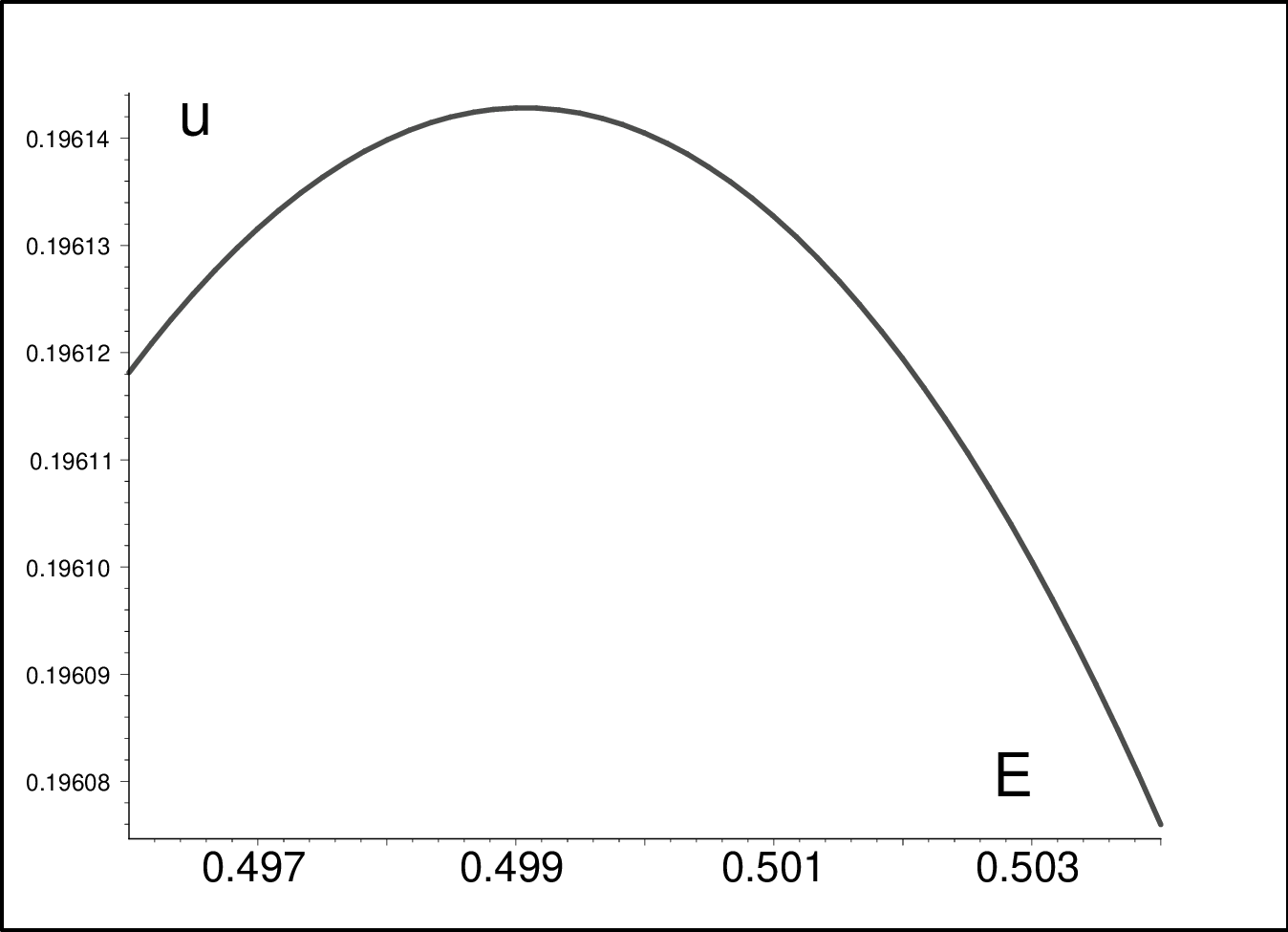,angle=270,width=0.3\textwidth}
\end{center}
\vspace{-2mm}\caption{$u(E)$ for $N=5$ near the right EP - magnified.
 \label{5cdrufi}}
\end{figure}

Outside of this interval (with the endpoints
representing the two EP singularities)
the $r=0$ spectrum becomes composed of the three real
and two complex eigenvalues.
After we return to the analytic approach
we obtain the following formula
for the Sturmian of relevance,
 \be
 u(E)=
 {\frac {1-3\,{{\it {E}}}^{2}+{{\it {E}}}^{4}-\sqrt {1-4\,{{
 \it {E}}}^{2}+4\,{{\it {E}}}^{4}-{{\it {E}}}^{6}}}{{{\it {E}}}^{3}-2\,{
 \it {E}}}}\,
 \ee
Again, the approximate, graphical search for the positions
of the EPs can be based on Figure \ref{5drufi}.
The validity of the approximation published in \cite{init}
is confirmed by
Figure~\ref{5cdrufi} which is just
the
magnified
version of
the relevant part of Figure~\ref{5drufi}
which is, by itself, just a
magnified
version of
the relevant part of Figure~\ref{5bdrufi}.

\section{\textcolor{black}{Beyond $N=5$}}

\textcolor{black}{After one compares, once more,
Figures \ref{3drufi} (where $N=3$ is odd)
and \ref{4drufi} (where $N=4$ is even)
one easily accepts an assumption that
having now, at our disposal, the analytic as well as numerical
characteristics of the $k=3$ model with $N=5$,
one can hardly expect the emergence of any surprise at $N=6$.
Obviously, much more exciting becomes the project of
the study of
the ``next$-k$'' model with $N=7$.}

\textcolor{black}{
\begin{figure}[h]
\begin{center}
\epsfig{file=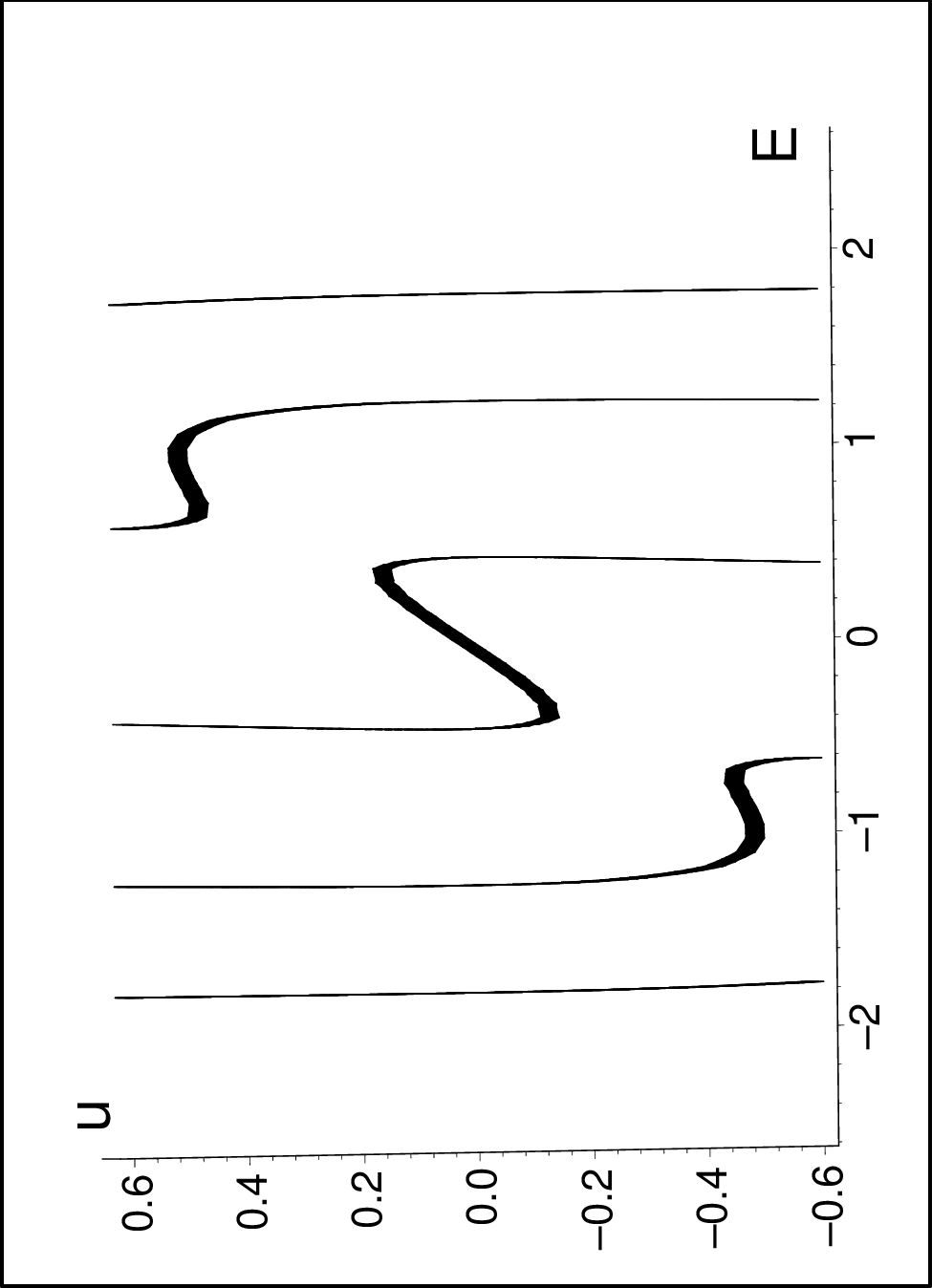,angle=270,width=0.3\textwidth}
\end{center}
\vspace{-2mm}\caption{$u(E)$ for $N=7$.
 \label{7drufi}}
\end{figure}}

\textcolor{black}{In some sense, the result of the $N=7$
calculations is truly surprising. Even though such a result could
have been given here, again, a closed and explicit analytic form
(after all, also this task may be left again to the readers),
a much more concise and persuasive message is being mediated and
provided by Figure \ref{7drufi}
in which we clearly see a decisive qualitative difference from
its $N=5$ predecessor of Figure \ref{5bdrufi}.}

\textcolor{black}{First of all, we notice that the number of the
EP degeneracies grew from two at $N=5$ to six at $N=7$. Secondly,
from a complementary point of view,
the picture clearly demonstrates that
the whole spectrum remains real (i.e., that
the evolution of the underlying quantum system
remains unitary) not only near $u=0$
(when the real part of parameter $z$ or
of function $z(t)$
in the Hamiltonian of Eq.~(\ref{Ka8t})
remains small)
but also inside the two small intervals
where the values of $u \approx \pm 0.46$ are safely non-vanishing.}

\textcolor{black}{
\begin{figure}[h]
\begin{center}
\epsfig{file=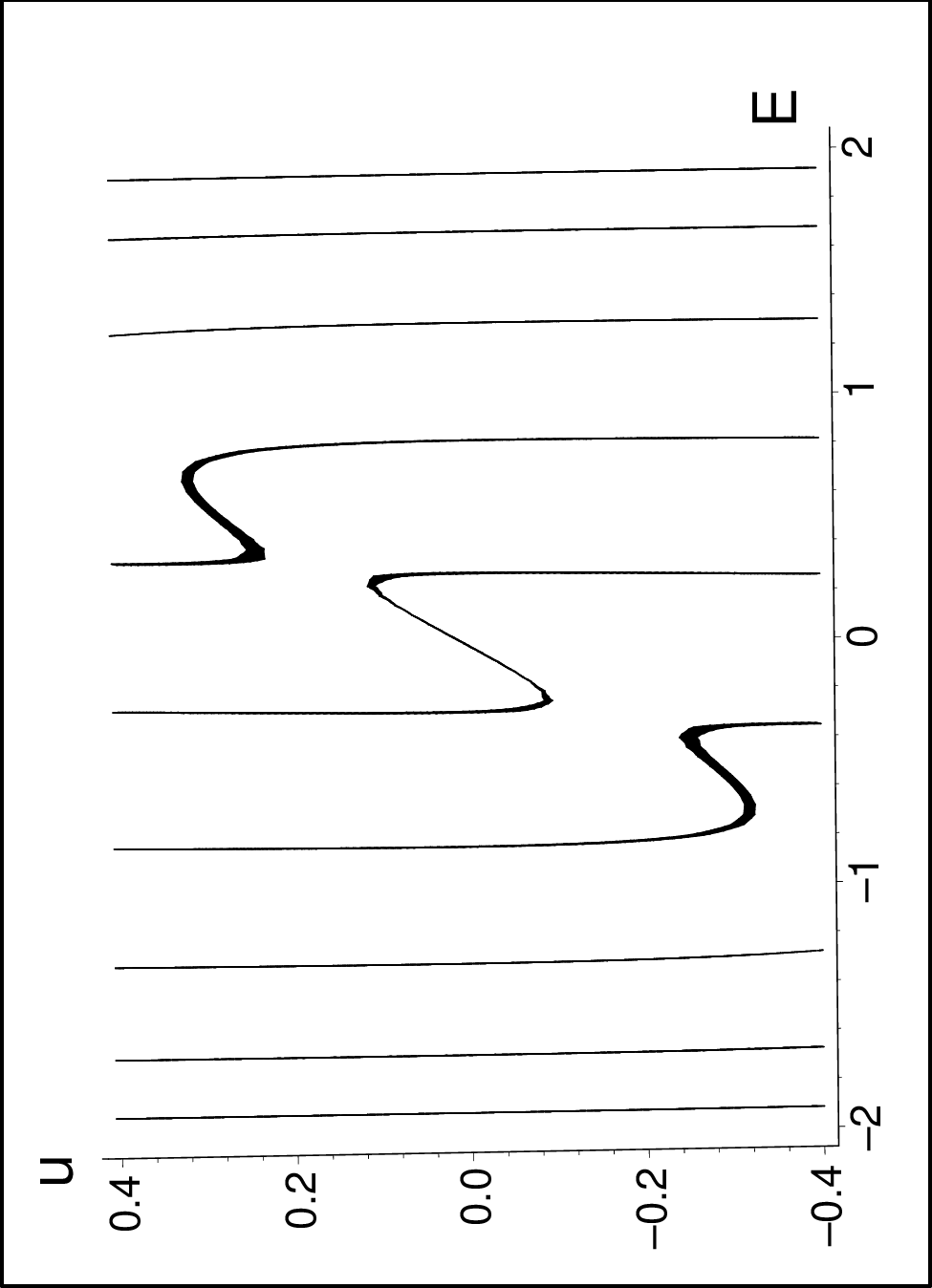,angle=270,width=0.3\textwidth}
\end{center}
\vspace{-2mm}\caption{$u(E)$ for $N=11$.
 \label{11drufi}}
\end{figure}}

\textcolor{black}{What can be also considered
remarkable is that our three
``intervals of unitarity''
are separated by the ``gaps of non-unitarity''
in which the spectrum ceases to be all real.
An apparent paradox is clarified easily because
the phenomenon
just reflects the fact that
the energy levels merging at the
EP boundaries of the separate intervals of $u$
are different.}

\textcolor{black}{In our last comment
on the phenomenon we have to add that virtually
all of the later features of the
$N=7$ model seem to be generic. Indeed,
we draw several
$k>4$ descendants of
Figure \ref{7drufi} (where $k=4$),
and we found that what is only added at $k>4$
are just the decoupled, ``outer observer'' energy levels:
At $k=6$ (i.e., at $N = 11$)
this is illustrated in our last Figure  \ref{11drufi}.}


\section{Discussion}

The basic idea of our present project
of the search for certain specific EP singularities
was twofold.
The first one
was theoretical. Its essence
can be seen in the
admissibility of quantum models
using, formally, non-Hermitian operators.
This, indeed,
extended the scope of the theory while opening
the possibility of control
of
the fate of classical singularities
after quantization.

\textcolor{black}{On the experimental physics side,
various experimental simulations have been performed recently,
ranging from rather elementary
coupled LRC circuits \cite{SSH}
and systems of ultracold atoms \cite{SSHb}
up to the truly sophisticated
coupled optical waveguides \cite{Christodoulides}, etc.
Still, in our present paper, our initial} idea
was purely pragmatic. Reflecting the conventional wisdom that
the essence of many puzzling technical questions
(emerging only during the practical implementations
of abstract considerations)
becomes fully clarified only when one tests the theory on
a sufficiently  simplified \textcolor{black}{schematic} toy model.

In the past, \textcolor{black}{the similar combinations of the
ambitious
theoretical considerations with
the equally ambitious
experiments and} observations
were accompanied by the
scepticism as expressed in our brief note \cite{init}.
We worked there with
several elementary illustrative examples
but we only managed to
describe the properties of the models
using just some brute-force numerical methods.

Our insight in the problem proved only amended when we
managed to
unify the ideas. We realized that
one of the decisive shortcomings of the current quantum theory
(predicting the absence of singularities after quantization)
has to be
seen in the comparatively less developed
techniques of working with non-Hermitian operators.
We imagined that the use
and non-numerical descriptions
of non-Hermitian solvable
models
\textcolor{black}{(cf., e.g., \cite{Nimrode})}
could really open the way towards
a synthesis of the theory
with its
sufficiently transparent interpretations.

In the related literature we noticed that
only too many singularities emerging in
classical physical systems
(with their most prominent sample being the
Einstein's theory of gravity and cosmology)
are widely believed to disappear
and get smeared out
after quantization.
In this sense, our main aim was a search of the models
in which
the solvability is combined
with the existence of the genuine quantum EP
degeneracies.

For a long time, a key obstruction excluding the
models (\ref{usKa8t}) of (\ref{Ka8t})
with a purely imaginary parameter $z$
from our consideration
was the absence of EPs at the odd matrix dimensions $N$.
We found the difference between
models with the respective even and odd $N$ puzzling.
Fortunately, what we had in mind was
just a more or less inessential
difference between the respective presence and absence
of the EP singularity at a central
part of the energy spectrum with $E=0$.
Thus, a broadening of the
perspective was a key to the ultimate decisive progress and success.

A correct insight into the mechanism of the
emergence of the EP singularity
has been achieved via a return to the
numerical tests as presented in our note \cite{init}.
This inspired us to
add a non-vanishing real part to $z$.
Thus,
in our present final resolution of the
puzzles as formulated in \cite{init}
we finally found a unified approach to the model at both the odd and even $N$.
We were able to conclude that irrespectively of the parity of $N$,
the quantum
singularities
supported by the model
have an entirely analogous structure
realized via the
genuine quantum Kato's EP singularities.

With the prominent example of
the quantized Big
Bang singularity
being, presumably, too complicated
for qualitative analysis at present, we restricted
our attention to a much narrower problem of the
emergence and construction of the non-Hermitian EP
degeneracies to
the most elementary boundary-controlled toy model
in which it was possible to
simulate the emergence and unfolding of the EP
singularity by
the purely analytic
non-numerical means.

This enabled us to conclude that
the intuitive perception
of existence of a singularity can be
also given a fully consistent
probabilistic quantum-theoretical background and interpretation.
Naturally, with such a possibility being clarified
on a toy-model level, one has to expect
that in the nearest future, the study of some more realistic models
might open a Pandora's box
of multiple new and difficult mathematical challenges.

Among them, it is already possible to mention the
currently well known
enormous sensitivity of the systems near EPs
to perturbations (cf. \cite{Trefethen,Viola}
or a few remarks in Appendix D)
as well as
all of the related deeper
conceptual, physical and phenomenological
questions as formulated and discussed in the related older as well as newer
literature (cf., e.g.,
\cite{catast,Christodoulides,Berry,Heiss,Heissb}).

\newpage

\section*{Appendix A. \textcolor{black}{A note on} non-Hermitian degeneracies}

The traditional studies of conceptual differences
between the classical and quantum physics
found, recently, an unexpected source of a new inspiration
in astrophysics. In particular,
the cosmological hypotheses based on the classical physics
were confronted, recently,
with their quantized descendants in which
the process of quantization has been interpreted as a
reason for a replacement of the
classical point-like singularities
(like, typically, Big Bang,
cf., e.g.,
\cite{Planck,Planckb} or \cite{Penrose})
by their ``smeared''
quantum descendants
(sampled by the so called Big Bounce, cf. \cite{Bounce,piech}).
In such a context, one of
the applicability goals of
our present
study of the possible mechanisms of the
non-Hermitian degeneracies
may be seen in the statement
that
such a regularization need not be necessary.

\textcolor{black}{From the point of view of mathematics,
our argumentation has been based on a rather detailed
study of a fairly schematic toy model. In this sense
we cannot pretend to be able to establish a real contact with
experimentalists. In particular, in the
above-mentioned context of present-day astrophysics, there are only
too many open and difficult questions to be answered
on both the theoretical and/or experimental
level \cite{Rovelli,Thiemann,BBzpet}.
At the same time,
several methodical aspects of these questions are currently finding
some experimentally supported
answers in multiple contexts
ranging, typically, from classical and quantum optics
\cite{Christodoulides,Berry,Mousse,EP3}
and statistical physics \cite{Joshua,NIPb}
up to
the area of
contemporary cosmology \cite{Nimrodc} or
condensed matter \cite{Dyson,Heiss,Dysonb,Rotter,EvaM,EPsa,Liu}
or nuclear physics \cite{Jensen,Ingrid} or
physics of hybrid systems \cite{denish}
or quantum field
theory  \cite{Fisher,BM,KGali,usulumb,Carlsbook}
or physics of nonlinear systems \cite{Nimroda}.}

\subsection*{A.1. Theoretical framework}

{\it A priori,}
the {above-mentioned}
trends towards a delocalization of
Big Bang due to quantization
are far from  surprising.
They are widely accepted
even
in elementary models in which
we only take into consideration
a highly schematic model of the Universe.
For example, we may
follow paper \cite{few}
and decide
to quantize just the
age-dependent  spatial grid points.
Even then,
one intuitively expects that the
sharp grid-point
eigenvalues
get smeared \cite{Rovelli}.

{As we already mentioned above,
a decisive amendment of such a strongly misleading paradigm
only occurred after people realized that
the conventional textbook postulate of Hermiticity
of all of the observables (say, $Q$)
in ${\cal H}_{(physical)}$
is strongly dependent on our tacit assumption that
the latter Hilbert space and, in particular, its inner-product metric
is/are fixed in advance. In this sense,
it was rather revolutionary when Dyson \cite{Dyson}
simply changed the paradigm. What he proposed was
a simplification of the inner product.
This, in effect,
converted his initial conventional choice of the
physical but strongly ``user-unfriendly'' Hilbert space
${\cal H}_{(physical)}$
into a manifestly unphysical but persuasively
calculation-friendlier alternative
${\cal H}_{(mathematical)}$.}

{Not too surprisingly, the latter
amendment of the formalism
(which is currently called
quasi-Hermitian quantum mechanics,
cf., e.g., its oldest review \cite{Geyer})
found innovative applications, first of all,
in the description of complicated
structures of systems in nuclear physics
where any technical simplifications may
have a truly decisive impact (cf., e.g.,
\cite{Jensen}).
At the same time,
the idea
of the inner-product control
did not find an immediate impact,
say,
in the context of quantum field theory.
It
only had to be rediscovered there
after Bender with coauthors
restricted
their attention to a subset of eligible
quantum observables
which were required to exhibit a
technically helpful auxiliary property
called, by these authors,
parity times time reversal symmetry {\it alias\,}
${\cal PT}-$symmetry
(cf. review \cite{Carl} for details). }

An enormous success of the introduction
of the concept of
${\cal PT}-$symmetry
in several branches of physics \cite{Christodoulides}
attracted also the attention of mathematicians.
More or less immediately
they revealed that such a concept
is in fact just a special case of the
Hermiticity
of the relevant operators
in Krein space (cf., e.g., \cite{Dieudonne,BG,Langer}).
In some sense, unfortunately,
these developments led to
a certain destabilization of the terminology,
especially when Mostafazadeh
decided to unify the conventions and
proposed
to give the theory another name
of pseudo-Hermitian quantum mechanics \cite{ali}.

\subsection*{A.2. Phenomenology behind  non-Hermitian degeneracies}

The non-Hermitian degeneracies
played, initially, just a
purely formal role
in perturbation theory:
From the point of view of
an abstract mathematical analysis,
such a form of ``exceptional point'' (EP) singularity has been studied in the Kato's
comprehensive monograph \cite{Kato}.

Later on, the role of the mathematical objects
found its ubiquitous role in several  branches of physics \cite{St,return,[2]}
including even the traditional theory or resonant (i.e., unstable) states
\cite{Feshbach,Nimrod}.

It is, perhaps, worth adding that the
special, strictly pairwise
complex mergers,
say, of certain energy eigenvalues,
 \be
 \lim_{t\to t^{(EP)}}(E_{n_1}(t)-E_{n_2}(t))=0\,.
\label{pairwb}
 \ee
can be also found
in the quantum theory
of anharmonic oscillators \cite{Wu2,Wu}
\textcolor{black}{(with a decisive
methodical relevance in quantum field theory \cite{pert})}.

In all of these contexts, a key technicality is that one gets rid of
the conventional Hermiticity
(say, of any suitable non-stationary and
$N$-by-$N$-matrix observable $Q^{(N)}(t)$)
which is weakened
to read
 \be
  {Q}_{}^{}(t)   
 \ \neq \ \left [{Q}_{}^{}(t)\right ]^\dagger \ \
  \ \ {\rm in} \ \ \ \ {\cal H}_{(mathematical)}
\ \neq \ {\cal H}_{(physical)}
  \,
  \label{txtc}
  \ee
(here
we dropped the superscript
$^{(N)}$ as redundant).
One only has to add a complementary quasi-Hermiticity  \cite{Geyer,Dieudonne}
requirement
 \be
 {Q}^\dagger(t)\,\Theta(t)=\Theta(t)\,{Q}(t)\,.
\label{quhe}
 \ee
Again, operator $\Theta(t)$ stands here for a correct physical
inner-product metric \cite{Geyer,Brody,Faria,timedep,SIGMA,Fring,Ju}
which is, naturally, ambiguous \cite{SI}.

\section*{Appendix B. Closed versus open systems}

In the introductory part of our paper \cite{init}
we had to point out
that
all of the quantum models which we took into account were
not only non-Hermitian (in the sense of being assigned
some non-Hermitian operators
representing some relevant observable quantities)
but also, at the same time, hiddenly Hermitian
{\it alias\,} quasi-Hermitian, with
the origin of this terminological ambiguity dating back to the
comprehensive 1992 review paper \cite{Geyer}
by Scholtz, Geyer and Hahne.

The scope of our present
continuation of presentation \cite{init} is broader,
requiring a more detailed
terminologically-oriented explanations:
For the sake of brevity let us consider only the
subcategory of the quantum
systems
possessing just bound states.

\subsection*{B.1. Closed systems and their unitary evolution}

In the context of
the so called quasi-Hermitian quantum mechanics
of review \cite{Geyer}
(cf. also its more recent
and more detailed presentation and explanation in \cite{ali}),
the quantum system under consideration is considered ``closed'',
i.e., stable and unitary in an appropriate
physical Hilbert space of states ${\cal H}_{(physical)}$.

The first comment to be added is that
besides the obligatory
requirement of the reality of the spectrum
as imposed upon every relevant operator
representing an observable,
the description of the
closed quantum system might
still remain ambiguous and incomplete without
a rather thorough clarification
and disambiguation of terminology.

One of the rather unfortunate related sources of
potential misunderstandings lies in the
widely accepted tacit convention that
within the closed-system quasi-Hermitian framework
we do not perform the necessary calculations
in
${\cal H}_{(physical)}$
(i.e., in the standard physical Hilbert space
of conventional textbooks)
but rather in its auxiliary, decisively
user-friendlier alternative
${\cal H}_{(mathematical)}$.

The latter space
is, admissibly, manifestly unphysical.
One of the most unpleasant consequences of
this purely technical shortcoming is that the
relationship between
${\cal H}_{(physica)}$ and
${\cal H}_{(mathematical)}$
is not always properly kept in mind:
%
Still, the clarification
of the puzzle is rather easily achieved
using an appropriate consequent notation (cf. a few
more detailed comments in \cite{init}).
In particular,
a minor nontrivial amendment of the notation conventions
can be recommended in connection
with the
``ketket'' abbreviation $|\psi\kkt :=\Theta\,|\psi\kt$
where the symbol $\Theta$ denotes the so called
physical inner product metric operator
(see
 \cite{SIGMA,jupi}).

%
%
%
%
%
%


%

\subsection*{B.2. Unstable, open quantum systems}

In the preceding paragraph we admitted just
the quantum systems in which the evolution remains unitary.
This is to be guaranteed by the
existence of an appropriate metric operator $\Theta$.
Still, the scope
of the theory can be broadened to admit
the absence of unitarity as encountered
in many models of unstable systems
called open quantum systems
\textcolor{black}{emerging, for example, in
nuclear physics \cite{Ingrid} or in
condensed-matter physics \cite{Rotter}.}

These systems are, typically,
characterized by the
influence of an ``environment'' leading to the
emergence of  certain unstable states called
resonances~\cite{Nimrod}.
There is no doubt that the emergence of
resonances is characteristic for many realistic
branches of quantum physics
including, typically, the description of the many-body
nuclear, atomic or molecular systems.
In opposite direction, a return to
unitarity
can be then perceived
as a mere recovery of stability,
the admissibility of which keeps the theory
compact and more universal.

One of the technical difficulties is only
encountered
on the purely mathematical level
because the loss of
the reality of the
eigenvalues
would make both their
(numerical) search and (experiment-related) interpretation
perceivably more difficult.
Indeed, whenever one would like to
communicate with experimentalists and, say, predict
the results of measurements, one should have to determine
the (this time, complex) eigenvalues as precisely as possible,
offering a really model-independent
way towards the related physics.

%
%
%

After all, it is well known that
in open systems
the
complexity of the eigenvalues is a
consequence of
the existence of some more or less unknown
environment.
This means that non-unitary models
can still be considered realistic.


\section*{Appendix C. Numerical constructions}

\subsection*{\textcolor{black}{C.1. Complex boundary conditions in square well}}

\textcolor{black}{In conventional
textbooks \cite{Messiah}
the abstract mathematical principles of quantum theory
are often
illustrated using the simplest possible
square-well Schr\"{o}dinger equation
 \be
 -\frac{d^2}{dx^2}\psi_n(x)=\varepsilon_n\psi_n(x)
  \,,\ \ \
   \psi_n(-L)= \psi_n(L)=0\,,
   \ \ \ \ n=0,1,\ldots
 \label{spojhambc}
 \ee
or, alternatively, its numerically motivated \cite{Acton} difference-equation
approximate form
 \be
 -\psi_n(x_{k-1})+2\,\psi_n(x_{k})-\psi_n(x_{k+1})=
 E_n^{(N)}\psi_n(x_{k})\,
 \label{diskrhambc}
 \ee
where  $ k = 1,
2,\ldots,N$ and $\psi_n(x_{0})=\psi_n(x_{N+1})=0\,$.}

\textcolor{black}{In quasi-Hermitian quantum mechanics one can either
use the stationary non-Hermitian version of
Schr\"{o}dinger picture \cite{Geyer} or its non-stationary
interaction picture generalization \cite{book,timedep,Fring,NIP}.
In both of these scenarios, one of the
most natural points of making the dynamics nontrivial
are the boundaries of the interval. At these points it is sufficient
to use the
Robin boundary conditions
 \be
  \psi(-L)= \frac{\rm i }{\alpha + {\rm i}\beta}
  \,\frac{d}{dx}\psi(-L)\,,\ \ \ \ \ \
  \psi(L)= \frac{\rm i }{\alpha - {\rm i}\beta}
  \,\frac{d}{dx}\psi(L)\,
   \label{bcdavid}
 \ee
(in Eq.~(\ref{spojhambc})) or
 \be
  \psi_n(x_{0})= \frac{\rm i }{\alpha + {\rm i}\beta}\,
  \left (
  \frac{\psi_n(x_{1})-\psi_n(x_{0})}{h}
  \right )\,,
  \ \ \ \ \ \
  \psi_n(x_{N+1})= \frac{\rm i }{\alpha - {\rm i}\beta}\,
  \left (
  \frac{\psi_n(x_{N+1})-\psi_n(x_{N})}{h}
  \right )\,
 \label{prvnice}
  \ee
(in Eq.~(\ref{diskrhambc})).}

\textcolor{black}{The main advantage of this constraint is
that it
contains
two
parameters $\alpha \,,\,
\beta\, \in \,\mathbb{R}$
which violate the
Hermiticity of the Hamiltonian
while still preserving the
reality of the bound-state-energy spectrum \cite{init}.
This makes the model (equivalent to the one with
matrix Hamiltonians (\ref{usKa8t}) or (\ref{Ka8t}))
suitable for various methodical purposes.}

\subsection*{\textcolor{black}{C.2. Vicinity of singularities}}

The task of a constructive study of the
properties of quantum systems near their
exceptional-point dynamical extremes
becomes particularly challenging
when the authors of such a study
try to combine the requirements of mathematical rigor
with the ambition of making some experimentally verifiable predictions.

In our present paper we separated these two
requirements.
For the purposes of mathematical insight we
used just the most elementary
operators of observables.
Still, even in our schematic, boundary-controlled square-well
pseudo-Hermitian models, the computer-assisted numerical
calculations appeared challenging (cf.
\cite{init}, with several further relevant references therein)
as well as useful:
They helped us to reveal
the
slightly counterintuitive nature of the non-Hermitian quantum
theory
in both of its stationary and non-stationary realizations.

In particular, we found that the latter
formal shortcoming of the theory
can be perceivably weakened during its
various specific
toy-model implementations.
In all of these implementations, what is shared
as a decisive advantage
is the fact that
in contrast to the textbook models with trivial
identity-operator metric $\Theta=I$,
the
non-Hermitian
systems are now allowed to reach their singularities.
In the purely numerical setting, nevertheless,
it is well known that
when we want to study the
properties of
systems near their
EP singularities,
the influence of the rounding errors rapidly increases
with the decrease of the distance of
the parameter from its EP value
(see
Table Nr.~1 in Ref.~\cite{EPnum}).

This observation was the very essence of the message as delivered
in \cite{init}.
For definiteness,
we restricted our attention there
to the two separate domains of applicability
of the idea of a consistent coexistence
of a singularity on both the classical and quantum-theory level.
In both cases
we paid attention just to the
quantum system, the states of which were defined
in a finite, $N-$dimensional
Hilbert space ${\cal H}^{(N)}_{physical}$.


In a mathematically oriented and
less phenomenologically ambitious
part of the message of paper \cite{init}
the observable characteristics of the
quantum system in question were assumed
represented
by a time-dependent
and very specific $N$ by $N$ matrix (\ref{usKa8t})
representing a toy model with boundary-controlled dynamics.

\section*{Appendix D. Quantum physics near the singularities}

In the conventional textbooks on quantum mechanics
it is usually pointed out that
the singularities emerging
in various classical
physical systems
get very often smeared out
after quantization.
In this context we believe that such a
``rule of thumb'' need not be universally valid.
The essence of our persuasion is that
there exist non-equivalent approaches to the process of quantization,
in the framework of at least some of
which the singularities
attributed to some classical physical system
(and described, often, by
the so called theory of catastrophes \cite{Zeeman})
can find a very natural singular quantum counterpart \cite{catast}.

\subsection*{D.1. The vicinity of
singularity after quantization}

In this Appendix our attention will be paid
to the circumstances of
the
emergence of the singularities in a genuine quantum
dynamical regime.
We have to emphasize that in their admissibility one can see
one of the main phenomenological advantages
of the models using non-Hermitian
operators of observables.
The point is that
in the models using conventional Hermitian operators,
the eigenvalues
exhibit a tendency towards repulsion.
The characteristic consequence is the well known
avoided-level-crossing phenomenon \cite{init}.

%
%

In several papers including
also our most recent concise conference contribution
\cite{init}
we claimed that the intuitive and widely accepted
implication
``observability $\Longrightarrow$ avoided crossing''
need not hold.
We felt inspired by several
quantum-gravity interpretations of Big Bang in cosmology,
by which the
classical initial Big Bang singularity
becomes regularized and converted,
after quantization, into a Big Bounce
(cf. also a broader comprehensive review
of literature in
dedicated  monograph~\cite{Rovelli}).

After the recent quasi-Hermitian
reformulation of quantum theory,
it became clear that the survival of the
singularities after quantization cannot be excluded.
The main reason is that the eigenvalues of a
quasi-Hermitian operator
have a counterintuitive tendency of mutual attraction.
In fact,
this makes the possibility of an
unavoided crossing, in quantum as well as classical physics,
ubiquitous \cite{Heiss}.
In
classical optics, for example,
the phenomenon is frequently observed and known
under an indicative
nickname of ``non-Hermitian degeneracy'' \cite{Berry}.

In
quantum theory, the
instant of the non-Hermitian-degeneracy
singularity is widely interpreted
as the Kato's ``exceptional point'' (EP, \cite{Kato}).
In practical model-building processes, unfortunately,
even the very proof of the existence
of the
exceptional-point singularity
is never too easy.
The support of EP is a
feature of the models which is extremely sensitive to
perturbations
(see \cite{Viola}).
In the language of mathematics, also this
observation contributed, significantly,
to the formulation of our present research project.

\subsection*{D.2. Singularities in non-stationary dynamical regime}

In the context of study of models (\ref{usKa8t})
the points we addressed in \cite{PS1}
were partly methodical
and partly model-specific.
On the methodical side
we cited the relevant literature and, in particular,
we recalled and used
our original generalized
formulation of a consistent non-stationary
generalization of
quasi-Hermitian quantum mechanics \cite{NIP}.

We
emphasized that
a key to the transition from stationary to non-stationary
formalism lies in the factorization
of the metric
into factors called Dyson maps \cite{Dyson},
 \be
 \Theta(t)=\Omega^\dagger(t)\,\Omega(t)\,.
 \label{fakof}
 \ee
In the case of our present manifestly
non-Hermitian and non-stationary
model (\ref{usKa8t}) with complex
and time-dependent $z=z(t)$,
an explicit realization of
factorization (\ref{fakof})
was also one of the main highlights in \cite{PS1}.
In comparison with the
stationary results of paper \cite{[17]}
the news were nontrivial.
The simplicity
of our quasi-Hermitian observable of Eq.~(\ref{usKa8t})
enabled us to list and review
all of the subtle consequences
of the combination of the non-Hermiticity with non-stationarity.

With the purely imaginary function $z(t)$
we were even able to illustrate the consequences
of the non-stationarity,
in an explicit algebraic manner,
in the first nontrivial special case with $N=2$.
These results were non-numerical,
involving
not only the
constructions of the non-stationary
matrices $\Theta(t)$ and $\Omega(t)$ but also
the decomposition
of our preselected ``observable Hamiltonian'' $H(t)$ of Eq.~(\ref{usKa8t})
into a superposition  $H(t)=G(t)+\Sigma(t)$
containing the  ``Schr\"{o}dinger Hamiltonian''
component $G(t)$ (i.e.,
the wave-function-evolution generator)
together with the  ``Heisenberg Hamiltonian''
component
{\it alias\,}
``quantum Coriolis force'' 
$\Sigma(t)$,
i.e., the
operator
which formally controls the
evolution
of any relevant observable of the system
via Heisenberg equation.

\subsection*{Funding:}

The research received no funding.

\subsection*{Data availability statement:}

No new data were created or analysed in this study.

\subsection*{Conflicts of Interest:}

The author declares no conflict of interests.

\subsection*{ORCID iD:}

Miloslav Znojil: https://orcid.org/0000-0001-6076-0093

\end{document}